\documentclass{article}
\usepackage{ijcai18}

\usepackage{times}
\usepackage{xcolor}
\usepackage{soul}
\usepackage[utf8]{inputenc}
\usepackage[small]{caption}
\usepackage{graphicx}
\usepackage{subfig}
\usepackage{pdfpages}

\usepackage{amsmath}
\usepackage{amssymb}
\usepackage{amsthm}

\usepackage[linesnumbered,boxed,ruled]{algorithm2e}
\usepackage{comment}

\newboolean{withAppendix}
\setboolean{withAppendix}{true}

\DeclareMathOperator*{\Exp}{\mathbb{E}}

\newcommand{\eu}{\bar{u}}
\newcommand{\euo}{\bar{u}^\text{OPT}}
\newcommand{\eup}{\bar{u}^\delta}
\newcommand{\loss}{l}

\newtheorem{theorem}{Theorem}
\newtheorem{lemma}{Lemma}
\newtheorem{proposition}{Proposition}

\newtheorem{definition}{Definition}
\newcommand{\sm}{\textrm{-}}
\newcommand{\smi}{{\sm i}}

\renewcommand{\epsilon}{\varepsilon}

\newcommand{\citet}[1]{\citeauthor{#1} [\citeyear{#1}]}

\let\OLDthebibliography\thebibliography
\renewcommand\thebibliography[1]{
  \OLDthebibliography{#1}
  \setlength{\parskip}{0.18pt}
  \setlength{\itemsep}{0pt plus 0.3ex}
}

\newcommand{\llggridsize}       {600}
\newcommand{\llggridsizeold}       {800} 
\newcommand{\llgMCsamples}      {20,000}
\newcommand{\llgspeedup}        {6.45} 
\newcommand{\llgBNEdistance}    {0.043}
\newcommand{\llgtargeteps}      {$0.001$}

\newcommand{\llllgggridinterval}   {$\frac{1}{32}$}
\newcommand{\llllgggridsizeL}   {33}
\newcommand{\llllgggridsizeG}   {65}
\newcommand{\llllggMCsamples}      {200,000}
\newcommand{\llllggspeedup}     {77}
\newcommand{\llllggtargeteps}   {$0.02$}

\title{Non-decreasing Payment Rules for Combinatorial Auctions}
\author{
Vitor Bosshard\textsuperscript{1} , Ye Wang\textsuperscript{2} \and Sven Seuken\textsuperscript{1}\\
\textsuperscript{1}Department of Informatics, University of Zurich\\
\textsuperscript{2}School of Computer and Communication Sciences, EPFL\\
bosshard@ifi.uzh.ch, ye.wang@epfl.ch, seuken@ifi.uzh.ch
}

\begin{document}

\maketitle

\begin{abstract}
\vspace{-0.25em}
    Combinatorial auctions are used to allocate resources in domains where bidders have complex preferences over bundles of goods. However, the behavior of bidders under different payment rules is not well understood, and there has been limited success in finding Bayes-Nash equilibria of such auctions due to the computational difficulties involved. In this paper, we introduce \emph{\mbox{non-decreasing}} payment rules. Under such a rule, the payment of a bidder cannot decrease when he increases his bid, which is a natural and desirable property. \emph{VCG\nobreakdash-\hspace{0pt}nearest}, the payment rule most commonly used in practice, violates this property and can thus be manipulated in surprising ways. In contrast, we show that many other payment rules are non-decreasing. We also show that a non-decreasing payment rule imposes a structure on the auction game that enables us to search for an approximate Bayes-Nash equilibrium much more efficiently than in the general case. Finally, we introduce the utility planes BNE algorithm, which exploits this structure and outperforms a state-of-the-art algorithm by multiple orders of magnitude.
\end{abstract}

\vspace{-1.5em}
\paragraph{Correction:}
Unfortunately, Proposition~\ref{lem:proportional} is wrong (thanks to Vangelis Markakis and Artem Tsikiridis for helping us spot the error). 
The rest of the paper is only mildly affected.
Please see the corrigendum on our website for details.
\vspace{-0.25em}

\section{Introduction}
\vspace{-0.25em}

Combinatorial auctions (CAs) are commonly used to allocate multiple, indivisible goods to multiple bidders.
CAs allow bidders to express complex preferences on the space of all bundles of goods, taking into account that goods can be complements or substitutes \cite{CramtonEtAl2006CombAuctions}.
CAs have found widespread use in practice, including for the sale of radio spectrum licenses \cite{Cramton2013SpectrumAuctionDesign}, for the procurement of industrial goods \cite{Sandholm2013largescale}, and for the allocation of TV ad slots \cite{goetzendorff2014core}.

Unfortunately, the incentive properties of CAs are not very well understood. One exception, of course, is the well-know VCG auction, which is strategyproof \cite{vickrey1961counterspeculation,clarke1971multipart,groves1973incentives}. However, using VCG in a CA domain can lead to undesirable outcomes, such as very low revenue \cite{ausubel2006lovely} or collusion by bidders \cite{DayMilgrom2008CoreSelectPackageAuctions}.

\subsection{Non-truthful Payment Rules}

Due to VCG's shortcomings in CAs, non-truthful payment rules are used in practice.
The most prominent among these are core-selecting payment rules \cite{DayMilgrom2008CoreSelectPackageAuctions}.
A rule is core-selecting if there is no group of bidders that is envious towards the auction winners. To achieve this, payments must lie in the \emph{core}, which is the set of payments where no coalition of bidders is  willing to pay more than what the auction charged the winners.
The core is a convex polytope and the \emph{minimum revenue core} is the face of the core where the total payment of all bidders is minimized.  The \emph{VCG-nearest} payment rule (or ``Quadratic''), the rule most often used in practice, selects the unique point in the minimum revenue core that minimizes the Euclidean distance to the VCG payment point \cite{DayCramton2012Quadratic}. To date, VCG-nearest has been used in over ten spectrum auctions worldwide, generating more than \$20 billion in revenues \cite{ausubel2017practical}. Core-selecting CAs have also been proposed for use in electricity markets \cite{karaca2018designing} and in sponsored search auctions \cite{hartline2018fast}

Unfortunately, there does not exist a strategyproof core-selecting payment rule \cite{Goeree2013OnTheImpossibilityOfCoreSelectingAuctions}, and relatively little is known about the incentives and strategic behavior of bidders when non-truthful payment rules like \emph{first-price} or \emph{core-selecting rules} are used in CAs. Large settings with many bidders and goods are intractable to analyze, both analytically and computationally. Therefore, core-selecting CAs have mostly been analyzed in stylized settings \cite{AusubelBaranov2013CoreSelectingAuctionsWithIncompleteInformation,ParkesKalagnanamEso2001Achieving}.
Even so, many ways are known in which bidders can manipulate core-selecting CAs in their favor \cite{gretschko2016strategic,beck2013incentives}.

\subsection{Background on BNE Algorithms}

We employ the standard \emph{Bayes-Nash equilibrium (BNE)} concept to analyze the incentives of bidders in a CA. A BNE, in contrast to the related \emph{Nash Equilibrium}, captures the fact that bidders have incomplete information about their rivals' preferences. Bidders have a belief about other bidders' valuations, in the form of random variables drawn from the \emph{value space}.
Bidders submit their own bids, chosen from their \emph{\mbox{action space}}, in response to the expected bids of others.

For the purpose of equilibrium computation, \cite{bosshard2017fastBNE,reeves2004computing} have argued that CAs should be modeled as infinite games, i.e., games with continuous value and action spaces. The argument is that it is actually more difficult to compute equilibria for finite games with a large number of states than for infinite games. However, the infinite setting is very challenging in its own way: the BNE solution concept imposes an equilibrium condition on every valuation of every bidder simultaneously; thus, special care must be taken to ensure the equilibrium condition is met for all possible valuations. While it is possible to apply heuristics in this step (e.g. checking the equilibrium condition at a finite sample of valuations), such heuristic equilibria are difficult to interpret in a standard game-theoretic way. In contrast, the BNEs in the infinite game model are readily interpretable.

The state of the art in this area is the work of \citet{bosshard2017fastBNE}, who describe a fast numerical algorithm for finding approximate BNEs (i.e. $\epsilon$-BNEs) in the infinite setting.
They also point out several pitfalls that should be avoided, under the label of the \emph{false precision problem}, to ensure that the result of a search for $\epsilon$-BNEs yields a correct $\epsilon$.

The combinatorial nature of the problem, together with the incomplete information setting, make BNE computation a very hard, almost intractable problem. Even state-of-the-art algorithms can only solve instances of limited size. By size we mean the number of bidders and goods included in the auction, as well as the number of bundles each bidder bids on.
The latter directly affects the dimensionality of bidders' strategies, because the bids on all different bundles must be jointly optimized by each bidder. Scalability in this dimension is thus particularly difficult.

Despite the high complexity of this problem, advances in algorithmic techniques for BNE computation can still be very beneficial.
First, being able to quickly solve small-to-medium sized instances makes it possible to solve many of them, enabling approaches such as an algorithmic search for payment rules with desirable properties \cite{lubin2018designing}.
Second, understanding small CA instances helps us develop a deeper understanding of the effects that drive strategic behavior in combinatorial auctions.
For instance, the LLG domain has been extensively studied \cite{Goeree2013OnTheImpossibilityOfCoreSelectingAuctions,beck2013incentives,AusubelBaranov2013CoreSelectingAuctionsWithIncompleteInformation,baranov2010exposure}, even though the only combinatorial interaction that arises is two local players needing to cooperate to jointly outbid the global player. There are many more interesting interactions that can already happen in domains such as LLLLGG \cite{bosshard2017fastBNE}, which have a size in-between very stylized/toy domains such as LLG and a fully fledged spectrum auction with hundreds of goods and dozens of bidders. With the advent of more efficient BNE algorithms, the possibility of systematically studying such domains is just beginning to open up.

\subsection{Overview of our Contributions}

In this paper, we aim to find new structural properties of CAs that help us better understand the strategic behavior of bidders in CAs and also allow us to design more effective BNE algorithms. This leads us to introduce and analyze \emph{non-decreasing} payment rules for CAs. Such payment rules cannot decrease a bidder's payment when he increases his bid, unless the allocation changes. In Section~\ref{sec:WI}, we introduce this property formally. We prove that VCG-nearest, the payment rule most commonly used in practice, is not non-decreasing. We also show that non-decreasing rules actually exist, by placing several well-known rules in this class, including strategyproof and core-selecting ones. In Section~\ref{sec:algorithm}, we introduce a new algorithm to compute $\epsilon$-BNEs that exploits the structure of non-decreasing payment rules. For this, we first show how piecewise constant strategies can be used as an auxiliary modeling step in the computation of $\epsilon$-BNEs.
Such strategies create a cell structure that helps us understand the combinatorial interactions when bids are multidimensional. Under a non-decreasing payment rule, this structure can be exploited to find the best response of a bidder to the strategies of all other bidders more efficiently. We assemble these ideas into the \emph{utility planes BNE algorithm} which computes $\epsilon$-BNEs for infinite games, while guaranteeing the correctness of $\epsilon$. In Section~\ref{sec:experiments}, we evaluate the runtime of our algorithm in two different CA domains. We show that it outperforms the state-of-the-art algorithm by multiple orders of magnitude.

\vspace{-0.01in}
\section{Formal Model}
\vspace{-0.02in}

\begin{table*}
    \centering
    \begin{tabular}{|l||ccc|c|c||ccc|c|c|}
    \hline
      & & bids & & VCG & VCG-nearest & & bids$'$ & & VCG$'$ & VCG-nearest$'$\\
      & $\{1\}$ & $\{2\}$ & $\{1,2\}$ & & & $\{1\}$ & $\{2\}$ & $\{1,2\}$ & &\\
    \hline
    \hline
    Bidder \#1 & $4^*$ &     &   & 2 & 3    & $4^*$ &&& 3 & 3.5\\
    Bidder \#2 &       &$4^*$& 5 & 2 & 3    & &$4^*$&  \textbf{7} & 2 & 2.5\\
    Bidder \#3 & 2     & 2   & 6 & 0 & 0    & 2 & 2 & 6 & 0 & 0\\

    \hline
    \end{tabular}
    \caption{Auction where 3 bidders are bidding on 2 goods. Winning bids are marked with an $^*$. Bidder 2 increases his bid on a non-winning bundle (marked in bold), which decreases his payment on his winning bundle, while the allocation remains constant.}
    \label{tab:quadcounter}
\end{table*}

In this section, we introduce CAs with continuous values and bids. We use the well-known independent private values (IPV) model. In contrast to other work, we explicitly handle the occurrence of ties. This is necessary because we deal with piecewise constant strategies and thus cannot assume that ties will occur with probability 0, as would be the case with strictly monotone strategies.

\vspace{-0.02in}
\subsection{Combinatorial Auctions}
\vspace{-0.02in}
A combinatorial auction (CA) is a mechanism used to sell a set $M = \{1, 2, \ldots, m \}$ of goods to a set $N = \{1, 2, \ldots, n \}$ of bidders.
For each bundle of goods $K \subseteq M$, each bidder $i$ has a value $v_i(K) \in \mathbb{R}_{\geq 0}$, and submits a (possibly non-truthful) bid $b_i(K)$.
We assume that each bidder only bids on a limited number of $r$ \emph{bundles of interest} (typically a true subset of all possible bundles).
For a fixed $r$, the bid $b_i$ can thus be represented by a point in the action space $\mathbb{R}_{\geq 0}^{r}$, with bids on all other bundles implicitly being $0$.
The bid profile $b = (b_1, \ldots, b_n)$ is the vector of all bids, and the bid profile of everyone except $i$ is denoted $b_{\smi}$.
The CA has an allocation rule $X(b)$ which always produces an efficient allocation: it maximizes reported social welfare (the sum of all winning bids), by solving what is known as the \emph{winner determination problem}.
Bids can be such that ties occur, i.e. multiple allocations are efficient.
Thus, $X(b)$ is a correspondence (set-valued function), with each $x \in X(b)$ having the same probability of being chosen.
$x_i$ denotes the bundle assigned to $i$ under allocation $x$, possibly the empty bundle.
The CA also has a payment rule $p(b, x)$ which is a function assigning a payment to each bidder based on the bid profile and allocation.
We assume that the payment rule satisfies individual rationality, i.e. $p_i(b, x) \leq b_i(x_i)$.
Each bidder has a quasi-linear utility function $u_i(v_i, b, x) = v_i(x_i) - p_i(b, x)$.
We also assume that the utility functions satisfy free disposal, i.e. a bidder always has weakly positive value for winning additional goods.\footnote{Formally, free disposal requires that for all $K \subseteq K' \subseteq M$, we have that $v_i(K) \leq v_i(K')$.}

\subsection{CAs as Bayesian Games}

We model the process of bidding in a CA as a Bayesian game.
Each bidder knows his own valuation $v_i$, but he only has probabilistic information (i.e. a prior) over each other bidder $j$'s valuation $v_j$, represented by the random variable $V_j$.
The joint prior $V = (V_1, \ldots, V_n)$ is common knowledge and consistent between bidders.
We assume that the $V_i$ are mutually independent.
Each bidder chooses a strategy $s_i$.
We assume that all strategies are pure, i.e. $s_i$ is a function mapping values to bids.
The expected utility of bidder $i$ with value $v_i$ when bidding $b_i$ is given by
\begin{equation}
    \eu_i(v_i, b_i) := \Exp_{b_{\smi}\sim s_{\smi}(V_{\smi})} \left[ \Exp_{x \in X(b)} \left[ v_i(x_i) - p_i(b,x) \right] \right],
    \label{eq:util}
\end{equation}
with the inner expectation corresponding to tie-breaking between efficient allocations.
The amount of utility that a bidder $i$ is ``leaving on the table'' when submitting bid $b_i$ (instead of bidding optimally) is called the utility loss, given by
\begin{equation}
    \loss_i(v_i, b_i) := \sup_{b_i' \in \mathbb{R}_{\geq 0}^{r}} \,\, \eu_i(v_i, b_i') - \eu_i(v_i, b_i).
    \label{eq:loss}
\end{equation}
The expected utility might be discontinuous and not have a maximum, which is why we take the supremum instead. Bidders are in an $\epsilon$-equilibrium when the utility loss is small for all possible valuations of all bidders:

\begin{definition}
    An \textbf{$\varepsilon$-Bayes-Nash equilibrium} ($\epsilon$-BNE) is a strategy profile $s^*$ such that
    $$\forall i \in N, \forall v_i \in V_i : \loss_i(v_i, s_i^*(v_i)) \leq \epsilon.$$
\end{definition}

\section{Non-decreasing Payment Rules}
\label{sec:WI}

\begin{table*}
    \centering
    \begin{tabular}{|l||c|c|c|c||c|c|c|}
    \hline
      & bundle of interest & bid & VCG & VCG-nearest & bid$'$ & VCG$'$ &VCG-nearest$'$\\
    \hline
    \hline
    Bidder \#1  & $\{1\}$     & 5 & 2 & $37 / 12$ & 5 & 1 & $36 / 12$ \\
    Bidder \#2  & $\{2\}$     & 5 & 0 & $16 / 12$ & 5 & 0 & $18 / 12$ \\
    Bidder \#3  & $\{3\}$     & 4 & 1 & $37 / 12$ & \textbf{5} & 1 & $36 / 12$ \\
    Bidder \#4  & $\{4\}$     & 1 & 0 & $\phantom{0} 7 / 12$ & 1 & 0 & $\phantom{0} 6 / 12$ \\
    Bidder \#5  & $\{5\}$     & 1 & 0 & $\phantom{0} 7 / 12$ & 1 & 0 & $\phantom{0} 6 / 12$ \\
    Bidder \#6  & $\{6\}$     & 1 & 0 & $10 / 12$ & 1 & 0 & $12 / 12$ \\
    Bidder \#7  & $\{1,2,4\}$ & 5 &   & $       $ & 5 &   & $       $ \\
    Bidder \#8  & $\{2,3,5\}$ & 5 &   & $       $ & 5 &   & $       $ \\
    Bidder \#9  & $\{1,3,6\}$ & 7 &   & $       $ & 7 &   & $       $ \\
    Bidder \#10 & $\{4,5,6\}$ & 2 &   & $       $ & 2 &   & $       $ \\
    Bidder \#11 & $\{2,3,4\}$ & 5 &   & $       $ & 5 &   & $       $ \\
    \hline
    \end{tabular}
    \caption{Auction where 11 single-minded bidders are bidding on 6 goods. We  show VCG and VCG-nearest payments for each bidder. If bidder 3 increases his bid (marked in bold), this decreases bidder 1's VCG payment, which in turn decreases bidder 3's VCG-nearest payment.}
    \label{tab:quadcounter1}
\end{table*}

In this section, we introduce \emph{non-decreasing} payment rules. Informally, when a bidder increases his bid under such a rule, his payment cannot decrease unless the increased bid causes a change in the allocation. If the allocation stays fixed, then charging a smaller payment for a higher bid is counter-intuitive. However, if the allocation changes, then it might very well be appropriate (in economic terms) to decrease the payment for a higher bid. For example, a bidder might currently be allocated the highly sought-after bundle $K$, and upon increasing his bid on a less-demanded bundle $K'$, his allocation changes from $K$ to $K'$ and his payment decreases.
\begin{definition}
    For any allocation $x$, let $\mathcal{B}_x$ be the set of bid profiles for which $x$ is efficient.
    The payment rule $p(b, x)$ is \textbf{non-decreasing at $\textbf{x}$} if, for all bidders $i$ and bid profiles $b, b' \in \mathcal{B}_x$ with $b_\smi = b'_\smi$, the following holds:
    $$b_i' \geq b_i  \,\, \Rightarrow \,\, p_i(b', x) \geq p_i(b, x).$$
    A payment rule $p(b, x)$ is \textbf{non-decreasing} if it is non-decreasing at all allocations $x$.
\end{definition}
As we will show in Section~\ref{sec:algorithm}, non-decreasing payment rules exhibit a particular structure that we can exploit in the design of BNE algorithms. However, this property is also desirable from an incentive point-of-view. Consider the strategyproof VCG payment rule, which is obviously non-decreasing, because under VCG, bidder $i$'s payment is independent of his bid (and thereby also non-decreasing). Coupled with the welfare-maximizing allocation rule, this gives bidders an incentive to report their true valuations under VCG which leads to high efficiency. In contrast, a rule that is not a non-decreasing rule can be manipulated in surprising ways.

\subsection{The VCG-nearest Rule}

To illustrate this, consider the VCG-nearest payment rule. One justification for using this rule that has been provided in the literature is that VCG-nearest produces outcomes in the minimum revenue core and maximizes incentives for truthful bidding \cite{Cramton2013SpectrumAuctionDesign,DayMilgrom2008CoreSelectPackageAuctions}. It is intuitive that minimizing the distance to VCG should reduce the ``residual incentive to misreport'' \cite{DayCramton2012Quadratic}, but the evidence for this is sparse. Indeed, there are other ways to choose a minimum revenue core payment, and there is some evidence suggesting that VCG-nearest may not be the best \cite{lubin2018designing,bichler2014spectrum}.

As it turns out, VCG-nearest is not a non-decreasing payment rule, violating this property in many different situations.
First, we discuss a simple example, which will also serve as proof for our claim.

\begin{proposition}
    VCG-nearest is not non-decreasing.
\end{proposition}

\begin{proof}
    Given by counterexample in Table~\ref{tab:quadcounter}, where bidder 2 decreases the payment on his winning bundle $\{2\}$ by increasing his bid on losing bundle $\{1,2\}$ from 5 to 7.
\end{proof}

Note that bidder 2 could also cause his payment to decrease by directly decreasing his bid on $\{2\}$.
In contrast, the manipulation shown in Table \ref{tab:quadcounter} involves an \emph{over-bid}, which can be particularly problematic for the auction process. The CCA, which is the auction format most often used in practice (e.g. for spectrum auctions) has a \emph{clock phase} of several rounds used for price discovery, followed by a final round in which a sealed-bid CA is used. There are activity rules in place that force bidders to be consistent in their revealed preferences: their demand for goods must decrease in response to prices rising over time \cite{Cramton2013SpectrumAuctionDesign}. The manipulation we show here allows bidder 2 to declare high demand for bundle $\{2\}$ while keeping his anticipated payment low. This suggests that there may be ways to circumvent the activity rule by strategically exaggerating one's demand, a topic which should be investigated further.

This first counterexample already shows ways in which the VCG-nearest payment rule can be problematic.
Unfortunately, there exist even more egregious cases.
For instance, a bidder who is single-minded (i.e. who only bids on a single bundle) can sometimes decrease his payment by bidding higher on his bundle, even when he is already winning. VCG-nearest is thus doing the direct opposite of what would be intuitive when receiving an economic signal of higher interest in a bundle.

Consider the situation in Table~\ref{tab:quadcounter1}, where 11 bidders are bidding for 6 goods.
We show two different bid profiles, along with the corresponding auction outcomes.
The only difference between the bid profiles is that bidder 3 has increased his bid from 4 to 5, which causes bidder 1's VCG payment to drop from 2 to 1.
In both situations, the winners are bidders 1-6, and each of the remaining bidders imposes a single core constraint.
The minimum-revenue core is a line segment that fulfills all five of these constraints with equality, yielding a total payment of 9.5.
However, bidder 3's payment under VCG-nearest decreases from $37/12$ to $36/12$.

\subsection{Other Core-Selecting Payment Rules}

We have already seen that VCG is non-decreasing, but it is not a core-selecting rule like VCG-nearest. Fortunately, many well-known core-selecting payment rules are also non-decreasing. This is straightforward to see for the \emph{first-price} payment rule, which is non-decreasing because $i$'s payment is directly proportional to the bid on the bundle he wins, and independent of his bids on other bundles.
We now show that two other core-selecting rules are non-decreasing as well.
The first one is the \emph{proportional} rule, where the payments are on the core boundary and proportional to the winning bids. The second is the \emph{proxy} rule, where we imagine a proxy agent bidding on behalf of each bidder in infinitesimal increments until a point in the core is hit \cite{AusubelBaranov2013CoreSelectingAuctionsWithIncompleteInformation}.
We provide formal definitions of these payment rules in Appendix A of the full version of this paper.

\begin{proposition}
    \label{lem:proportional}
    The proportional and proxy rules are non-decreasing.
\end{proposition}
\paragraph{Correction:}
Proposition~\ref{lem:proportional} is wrong.
Please see the corrigendum on our website for details.
\begin{proof}
    We show this for proportional, the argument is analogous for proxy.
    Let $b$ and $b'$ be identical bid profiles, except for $b_i'(K) > b_i(K)$, and let $x$ be an allocation that is efficient under both $b$ and $b'$.
    There are three cases:
    (1) $x_i = \emptyset$, thus $i$'s payment is $0$ by individual rationality.
    (2) $x_i = K$.
    The core is the same w.r.t. $b$ and $b'$.
    Let $\tilde{p}$ be the unique point on the line of payments proportional according to $b'$ with $\tilde{p}_i = p_i(b, x)$. Since $\tilde{p}_j \leq p_j(b,x) \, \forall j \neq i$, $\tilde{p}$ lies weakly below the core, thus $p_i(b',x) \geq \tilde{p}_i = p_i(b,x)$.
    (3) $x_i = K' \neq K$.
    Each core constraint under $b'$ is weakly higher than under $b$, but payments have the same proportion in both cases, so $p(b,x)$ lies weakly below the core w.r.t. $b'$.
    Thus $p_i(b',x) \geq p_i(b,x)$.
\end{proof}

The intuition behind this proof is simple: Many core-selecting payment rules (such as proportional, proxy and vcg-nearest) can be defined through the combination of a reference point and a method to project the payment to the core from the reference point \cite{lubin2018designing}. The core constraints themselves always increase in response to higher bids, and the reference point used for proxy and proportional is the origin. It is therefore sufficient to show that the projection method is non-decreasing, which is what the proof does.

The existence of payment rules which are both core-selecting and non-decreasing suggests re-evaluating the use of VCG-nearest in practice. Another payment rule may bring with it all the advantages of being in the core while being less manipulable and easier to work with for bidders and auctioneers. We leave a more detailed exploration of the space of non-decreasing core-selecting payment rules for future work.

\section{The Utility Planes BNE Algorithm}
\label{sec:algorithm}

In this section, we present our \emph{utility planes BNE algorithm}, a highly efficient BNE algorithm for CAs with non-decreasing payment rules. We will make use of several building blocks. First, and most importantly, we will use an algorithmic trick, whereby our algorithm internally only considers \emph{piecewise constant} strategies, i.e., strategies that consist of a finite number of flat segments. This will imply that only finitely many different bids can occur. This allows us to compute approximate best responses in a highly efficient manner using \emph{utility planes}, as we will explain in Section~\ref{sec:utilplanes}. Note that if such a discretization were applied naively, this could lead to a phenomenon called the \emph{false precision problem} \cite{bosshard2017fastBNE}. This problem arises when an auction game is simplified to make it tractable, an equilibrium is computed in the simplified game, and is then translated back to the original (richer) game without justification. This would lead to underestimating the potential utility loss and thus the $\epsilon$ in the computed equilibrium. Fortunately, the fact that we have non-decreasing payment rules allows us to use a powerful structural result (given as Theorem~\ref{thm:utilbound}) to bound the utility loss, not only in relation to the restricted space of piecewise constant strategies, but with respect to \emph{all possible strategies}. We provide details of this in Section~\ref{sec:structure}.
This enables us to avoid the false precision problem and compute a true $\epsilon$-BNE of the original game. Our full BNE algorithm is presented in detail in Section~\ref{sec:algorithmdetails}.

\subsection{Iterative Best Response Algorithms}

Our BNE algorithm is based on the iterative best response paradigm.
This is a well known approach to finding equilibria of games \cite{brown1951iterative}. Starting at some strategy profile $s$, we repeatedly replace $s$ with a \emph{best response} to $s$, i.e. a strategy profile $s'$ where for each bidder $i$, strategy $s'_i$ maximizes his utility against the previous strategies $s_{\smi}$. While such algorithms usually have no convergence guarantees, they work well in practice and are robust to issues such as only having access to approximate best responses.

\subsection{Computing Best Responses via Utility Planes}
\label{sec:utilplanes}

\begin{figure}
\centering
\includegraphics[width=70mm]{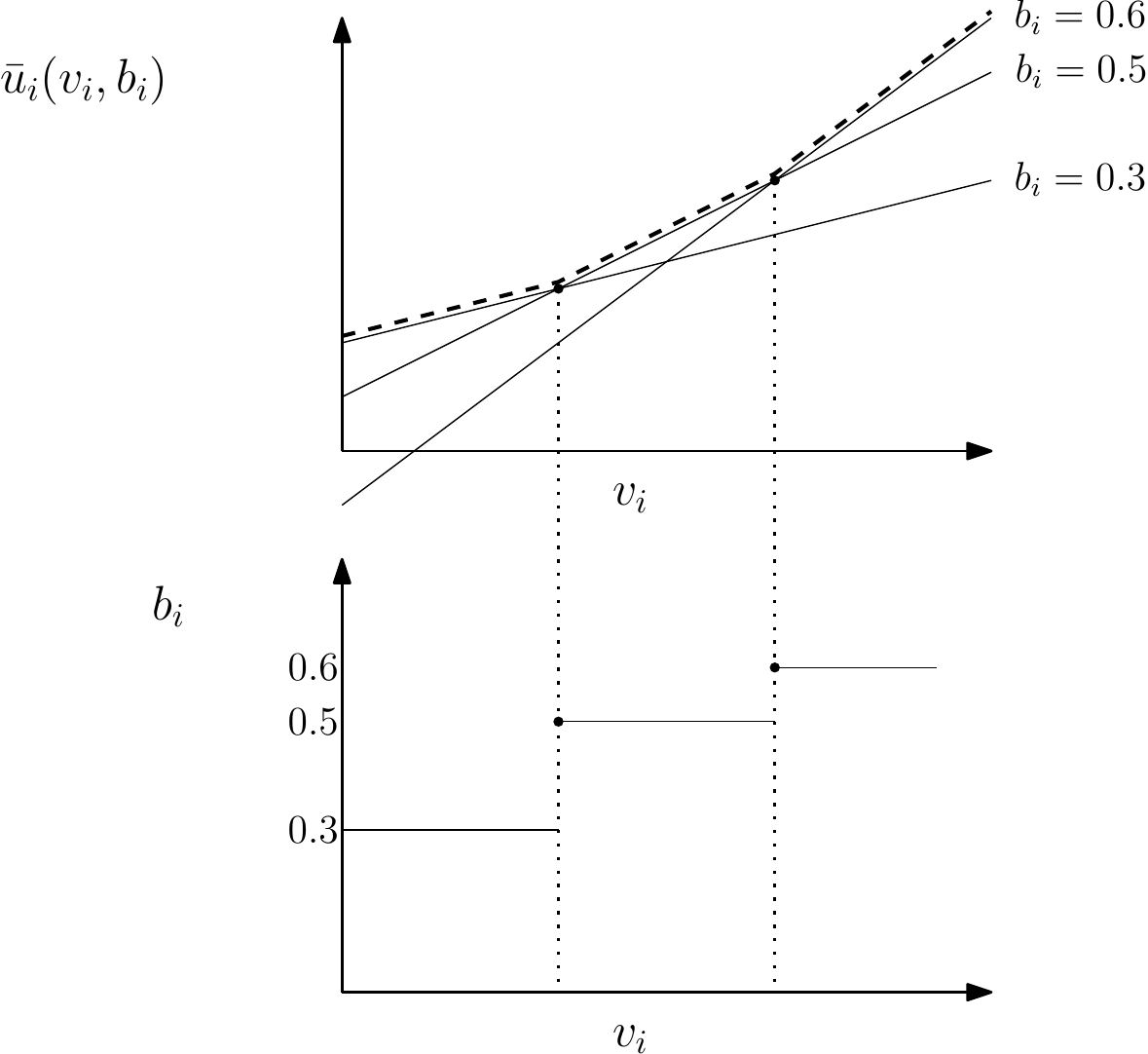}
\caption{Example of utility planes in one dimension. Top: Each possible bid of a bidder generates a utility plane (a line in this case). The upper envelope is marked as a dashed line. Bottom: Piecewise constant strategy induced by the set of utility planes.}
\label{fig:utilplanes}
\vspace{-0.03in}
\end{figure}

Recall that a strategy is a function mapping every valuation $v_i$ to a corresponding bid $b_i = s_i(v_i)$.
If we were to construct a best response to a strategy profile $s$ naively, we would need to find a bid that maximizes the expected utility for each valuation $v_i$ separately.
In fact, this is how \citet{bosshard2017fastBNE} compute best responses.
They find the optimal bid for a finite number of valuations and then interpolate between these \emph{pointwise best responses} to extend the strategy to the entire value space of the bidder.
This ``local'' approach is straightforward but computationally expensive, because the expected utility must be computed for many (value, bid) pairs, without sharing any intermediate results between them.

Fortunately, two of our assumptions (quasi-linear utilities and independently distributed valuations) allow us to do better: we use a ``global'' approach, where we construct the best response strategy $s_i'$ without ever directly computing $s_i'(v_i)$ for any particular $v_i$. Note that for a fixed bid $b_i$, the expected utility $\eu_i(v_i, b_i)$ is linear in $v_i$.
Specifically, $\eu_i$ contains one linear term for each bundle $K$, namely the value $v_i(K)$ times the probability of winning $K$.
The remaining term is the expectation of $-p_i(b, x)$ across all bundles, which is constant in $v_i$.
Thus, $\eu_i(v_i, b_i)$ forms a \emph{utility plane} that maps valuations $v_i$ to expected utility \cite{Rabinovich2013ComputingBNEs}.

For a set of such utility planes,
we can construct their \emph{upper envelope}, which is a piecewise linear function defined as the pointwise maximum of all these planes (Figure~\ref{fig:utilplanes}, top).
The value space is split into regions depending on which plane is the topmost one.
This upper envelope naturally induces the best response, which is a piecewise constant strategy: for each region in the value space, the bid $b_i$ associated to the topmost plane is the one maximizing $i$'s expected utility and is thus part of the best response (Figure~\ref{fig:utilplanes}, bottom).

The problem of computing upper envelopes is well known in computational geometry \cite{edelsbrunner1989upper}.
In our implementation, we avoid constructing the envelope explicitly, because this is expensive and fraught with numerical precision issues.
Instead, we evaluate the upper envelope for a finite set of valuations on a grid, and keep the strategy constant between grid points.
This makes the best response less accurate, but only in a negligible way.

\subsection{Partitioning the Action Space}
\label{sec:structure}

\begin{figure}
\centering
\includegraphics[width=80mm]{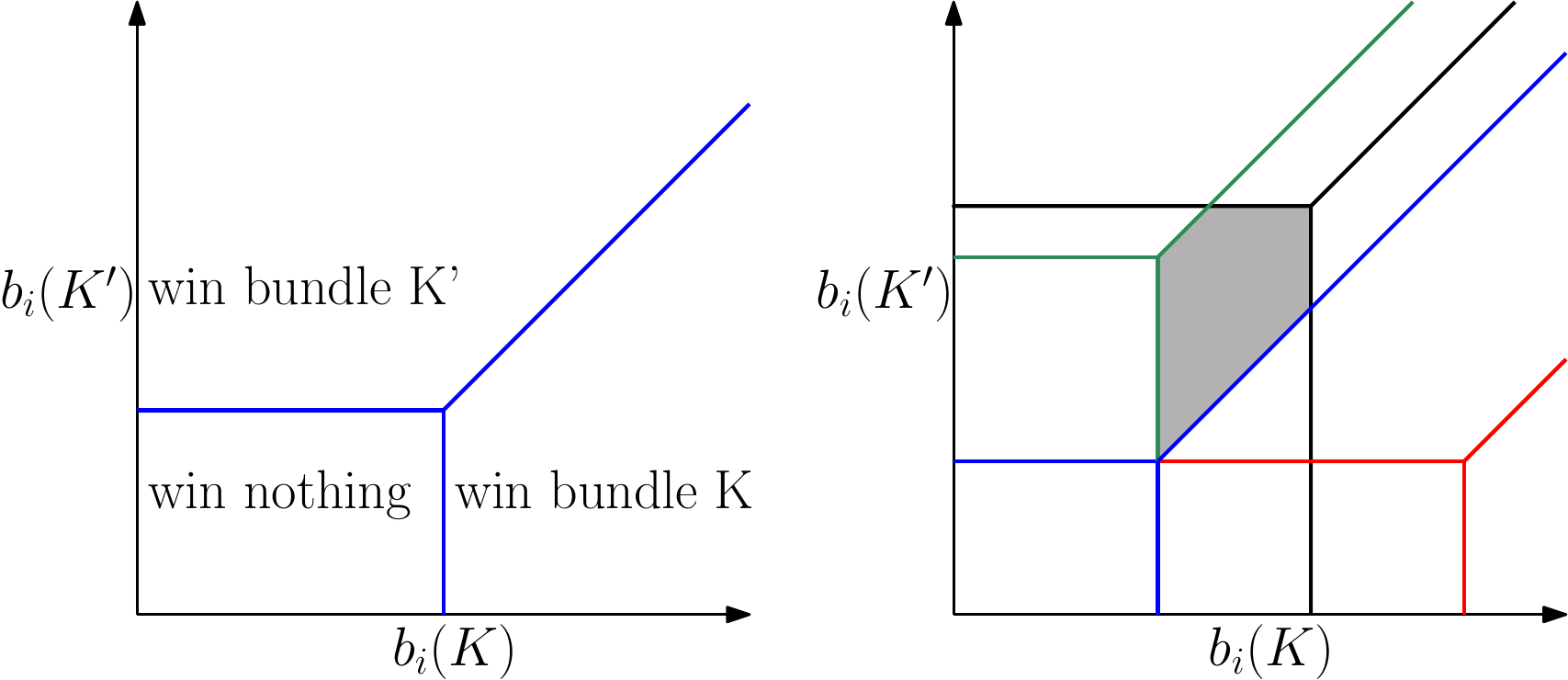}
\caption{Action space of bidder $i$, with two bundles of interest $K$ and $K'$. Left: for fixed bids $b_{\smi}$ of other bidders, the action space is partitioned into three convex polytopes, depending on which bundle bidder $i$ wins. Right: For a probability distribution over finitely many bids $b_{\smi}$ (represented here by black, green, red and blue), the action space is partitioned into cells where the probability of winning each bundle is the same within each cell. For instance, the cell highlighted in gray consists of all bids where $i$ wins nothing in the black case, bundle $K$ in the green case, and bundle $K'$ in the red and blue cases.}
\label{fig:cells}
\vspace{-0.03in}
\end{figure}

Recall that the action space of bidder $i$ is the set of all bids that $i$ is allowed to submit, i.e. $\mathbb{R}_{\geq 0}^{r}$, where $r$ is the number of bundles of interest.
We want to impose a structure on this space that will help us understand which bids we need to consider when computing the best response of $i$ to some strategy profile $s$.
This will also help us determine if $s$ is an $\varepsilon$-BNE.

Consider the situation where the bids $b_{\smi}$ of all other bidders are fixed.
In this case, $i$'s action space is partitioned into $r + 1$ convex polytopes, with $i$'s allocation being different in each polytope  (Figure~\ref{fig:cells}, left).
To explain how this partition arises, we need the concept of a \emph{constrained allocation rule}.
Recall that the allocation rule $X(b)$ maximizes reported social welfare across all bidders.
The constrained allocation rule $X_{\smi}(K, b_\smi)$ also maximizes welfare, but is subject to the constraint that bundle $K$ is allocated to bidder $i$.
Thus, the total reported social welfare achieved by $X_{\smi}(K, b_\smi)$ is a linear function of $b_i(K)$. Now, consider two bundles $K$ and $K'$. We can express the condition that $X_{\smi}(K, b_\smi)$ achieves higher reported social welfare than $X_{\smi}(K', b_\smi)$ as a linear inequality.
Given a bid $b_i$ where $i$ wins bundle $K$ under the (unconstrained) efficient allocation, such a linear inequality must hold for every bundle $K' \neq K$.
Thus, the set of all such bids forms a convex polytope.

Taking this idea further, if we have a probability distribution over finitely many bid profiles $b_{\smi}$ (i.e. when each $s_j$ is a piecewise constant strategy), we can partition the action space of $i$ into many smaller convex polytopes, such that inside each polytope, bidder $i$'s probability of winning each bundle is the same (Figure~\ref{fig:cells}, right). This finer partition can be visualized as intersecting the partitions corresponding to each possible $b_{\smi}$.
We call each of these smaller polytopes a \emph{cell} of bidder $i$'s action space.

Next, we show that for any strategy profile $s$, we can bound the maximum utility that bidder $i$ could possibly obtain by deviating from his current strategy.
To compute this bound, we only need to consider what happens at all cell vertices.
The idea is the following: if we have two different points $b_i \leq b_i'$ that are both strictly inside the same cell, then bidding $b_i$ is always preferable.
This is because the non-decreasing payment rule ensures that the expected payment at $b_i$ is weakly lower than at $b_i'$, while keeping the same distribution over allocations.
Taking this idea to its logical conclusion, it is never optimal to bid anything other than the lowest vertex of a cell.

There is a catch, however: a vertex itself belongs to several cells, so bidding exactly on the vertex would lead to ties between different allocations for bidder $i$, making the auction result different from that of inner points of the cell.
The optimal bid that is representative of a cell is thus strictly in the interior of the cell, just above the lowest vertex.

We cannot determine the payment for such an optimal bid exactly since we make no assumptions about the payment rule except that it is non-decreasing.\footnote{In particular, the payment rule need not be continuous.}
To solve this issue, we can compute the allocation as if $i$ had bid sligthly above the vertex, but collect the payment as if he had bid exactly on the vertex.
This results in an upper bound for the maximum utility that can be obtained anywhere in the cell in question.
Formalizing the above, we introduce the \emph{tie-winning expected utility}
\begin{equation}
    \eup_i(v_i, b_i, \pi) := \Exp_{b_{\smi}} \left[ \Exp_{x \in X(b_i + \delta\pi, b_{\smi})} \hspace{-1.5pt} \left[ v_i(x_i) - p_i(b,x) \right] \right],
\end{equation}
where $\delta$ is a small constant, and the vector $\pi$ corresponds to a permutation of $(1, \ldots, r)$ which moves the bid into one of the cells neighboring $b_i$. The permutation $\pi$ expresses some preference order over bundles in case of self-ties.
Next, we introduce the \emph{$i$-optimal expected utility}
\begin{equation}
    \euo_i(v_i, b_i) := \max_{\pi \in \Pi} \left[ \lim_{\delta \rightarrow 0} \, \eup_i(v_i, b_i, \pi) \right],
\end{equation}
which maximizes over $\Pi$, the set of all such permutations.
With this definition in place, we can state the main technical result of the paper.
\begin{theorem}
    \label{thm:utilbound}
    In a CA with a non-decreasing payment rule and where bidders have piecewise constant strategies, it holds that for every $v_i$, there exists a cell vertex $b_i^*$ such that, for any bid $b_{i}$ in $i$'s action space,
    $$\euo_i(v_i, b_i^*) \geq \eu_i(v_i, b_i).$$
\end{theorem}
\begin{proof}
    See Appendix B of the full version of this paper.
\end{proof}

\subsection{The Algorithm}
\label{sec:algorithmdetails}

With Theorem~\ref{thm:utilbound}, we now have all the pieces in place to build our algorithm.
On a high level, our algorithm works as follows:
We have a strategy profile $s$ consisting of piecewise constant strategies.
Thus, bidder $i$'s action space has finitely many cell vertices and we can construct best response $s_i'$ from the utility planes corresponding to these vertices. $s'_i$ is an approximate best response over the \emph{full action space} of $i$. The strategy profile $s'$ we construct in this way also consists of piecewise constant strategies, so we can repeat this procedure in the next iteration of our iterative best response algorithm.

\begin{algorithm}[tb]
	\KwIn{Auction $(X,p)$, valuations $V$, target $\widetilde{\epsilon}$}
	\KwOut{strategy profile $s$ and overall utility loss $\epsilon$}
	$s^0 = $ Initial strategies\\
	$\text{best\_iteration} = 0$ \\
	\For {$t = 1 \ldots \text{maximum\_iterations}$}
	{
	    $\epsilon^t = 0$
	
		\ForEach {bidder $i$}
		{
			\ForEach {Vertex $b_i$ in bidder $i$'s action space}
			{
				Compute utility plane of $b_i$ w.r.t. $s_{\smi}^{t-1}$, both with random and $i$-optimal tie breaking
			}
			$env_i^t = $ upper envelope with random tie breaking
			
			$env_i^{t,\text{OPT}} \hspace{-2pt} = $ upper envelope with $i$-optimal tie breaking
			
			$s_i^t= $ strategy induced by $env_i^t$
			
			$s_i^{t,\text{OPT}} = $ strategy induced by $env_i^{t,\text{OPT}}$
			
			$Q_i = $ vertices of $env_i^{t, \text{OPT}}$  $\cup$ vertices of $env_i^{t-1}$
			
			\ForEach {$v_i \in Q_i$}
			{
			    $l_i = \euo_i(v_i, s_i^{t,\text{OPT}}(v_i)) - \eu_i(v_i, s_i^{t-1}(v_i))$
			
			    $\epsilon^t = \max \left[ \epsilon^t, l_i \right]$
			}

		}

		\If{$\epsilon^t < \epsilon^{\text{best\_iteration}}$}
		{$\text{best\_iteration} = t$}
	}
	return $(s^{\text{best\_iteration}}, \epsilon^{\text{best\_iteration}})$
	\caption{Utility Planes BNE Algorithm}
	\label{alg:ibr}
	
\vspace{-0.03in}
\end{algorithm}

Our full BNE algorithm is provided in Algorithm~\ref{alg:ibr}.
The iterative best response loop is given in lines 3-22.
In lines 6-8, we compute two sets of utility planes for each bidder: one for the actual expected utility $\eu_i$, and another for the $i$-optimal expected utility $\euo_i$.
In lines 9-10, we compute two upper envelopes, one for each set of utility planes.
Line 11 computes bidder $i$'s best response, i.e. the strategy obtaining the highest possible utility in response to $s_\smi^{t-1}$.
Line 12 computes the best response, assuming that $i$ wins all ties and obtains the $i$-optimal utility.
The utility loss of bidder $i$ is calculated in lines 13-17 by comparing the utility obtained by the current strategy $s_i^{t-1}$ to the utility obtained by the $i$-optimal best response at all vertices of both envelopes.
In lines 19-21, we check if the strategy profile of the last iteration is a more accurate BNE than any previous strategy profile. The next theorem shows that Algorithm ~\ref{alg:ibr} computes a true $\epsilon$-BNE:

\begin{theorem}
    In a CA with a non-decreasing payment rule, Algorithm~\ref{alg:ibr} returns strategy $s$ and $\epsilon$, such that $s$ is an $\epsilon$-BNE.
    \label{thm:algoiscorrect}
\end{theorem}
\begin{proof}
    It follows from Theorem~\ref{thm:utilbound} that, for all $v_i$, the envelope $env_i^{t, \text{OPT}}(v_i)$ is an upper bound for the utility obtainable at $v_i$ with any bid $b_i$.
    For any two piecewise linear functions $f$ and $g$, the maximum of the function $h := g - f$ must occur at a vertex (one of the segment points) of either $f$ or $g$. Both $env_i^{t, \text{OPT}}$ and $env_i^t$ are piecewise linear functions, and if we subtract them we get an upper bound for the utility loss. It follows that $\epsilon_i$ computed in lines 14-17 of Algorithm~\ref{alg:ibr} bounds the maximum utility loss of bidder $i$.
\end{proof}
\paragraph{Discussion.}
With this theorem, our results are on solid ground:
While our algorithm only considers a finite set of bids (because internally, it only operates with piecewise constant strategies), it measures the $\epsilon$ it achieves against the highest utility achievable by \emph{any } bid, thus avoiding the false precision problem described by \citet{bosshard2017fastBNE}.

The reason why our algorithm works well in practice also becomes clear now: the best responses we compute are almost optimal.
In fact, a bidder loses very little utility by being restricted to bids that are cell vertices, namely the utility lost due to ties.
This amount is very small when the grid resolution is large enough and value distributions are relatively close to uniform.

\subsection{Runtime Analysis}

Computing a utility plane is by far the most expensive part of Algorithm~\ref{alg:ibr}.
Therefore, we evaluate its runtime in the unit cost model, counting the number of times this operation must be performed, and ignoring any other operations (e.g. the computation of the upper envelope).

The algorithm computes one utility plane for each vertex $b_i$ of the cell structure of bidder $i$.
Luckily, the number of such vertices is bounded. If all strategies are piecewise constant, then we have a finite number of distinct bids used by each bidder. Thus, we can always find a step size such that any bid by any bidder on any bundle is an integer multiple of this step size. It turns out that if we build a regular grid with resolution equal to this step size and lay it over the action space of a bidder, then all cell vertices fall exactly on grid points.
\begin{lemma}
    In a CA where bidders have piecewise constant strategies, if it is always the case that for some constant step size $c \in \mathbb{R}_{\geq 0}$, $b_i(K)$ is an integer multipe of $c$, then the coordinates of all cell vertices are also integer multiples of $c$.
\label{lemma:vertexgrid}
\end{lemma}
\begin{proof}
    Each cell vertex $b_i$ is defined by a linear system of equalities.
    The coefficients of these equalities are always $1$, $0$ or $-1$, so if we perform Gaussian elimination to find $b_i$, we will add and subtract equations, but never multiply them.
    The claim easily follows from the fact that $\{a \cdot c : a \in \mathbb{Z}, c \in \mathbb{R}\}$ is closed under addition.
\end{proof}
Lemma~\ref{lemma:vertexgrid} allows us to bound the amount of work the algorithm does per iteration.
If we choose the initial strategies to only contain bids on a regular grid, then the number of cell vertices does not grow over time, assuming that bidders never want to bid above their maximum value for any bundle.\footnote{For the first-price, proxy and proportional payment rules, we have found this to be the case in our experiments. Proving this property formally is an interesting open problem.}

For instance, if we have a bidder interested in two bundles with values in the $[0, 2]$ interval, choosing a step size of $c = 1/32$ means that we need to compute at most $65^2$ utility lines in each iteration.
More generally, if a bidder has $r$ bundles of interest, with each bundle $K_j$ having value in the interval $[0, v_i^\text{max}(K_j)]$, then all vertices lie on an $r$-dimensional grid with resolution $c$, going from 0 to $v_i^\text{max}(K_j)$ in each dimension.
In that case we need to compute at most $(\frac{1}{c}+1)^r \cdot \prod_{j=1}^r v_i^\text{max}(K_j)$ utility lines for that bidder.

\section{Experimental Results}
\label{sec:experiments}

In this section, we evaluate the performance of our utility planes BNE algorithm. We benchmark it against the algorithm from \citet{bosshard2017fastBNE} in two domains: the single-minded LLG domain and the multi-minded LLLLGG domain.

\subsection{Experiment Setup}

For each domain, we set a target $\epsilon$.
This target is \llgtargeteps{} for LLG, and \llllggtargeteps{} for LLLLGG.
For each algorithm, we measure the runtime required to reach a strategy profile that is proven to be an $\epsilon$-BNE.
The starting strategy profile is one where all bidders bid truthfully.\footnote{The cell structure we describe in Section~\ref{sec:structure} only arises when all bidders play piecewise constant strategies.
Therefore, our BNE algorithm computes an inaccurate $\epsilon$ in the first iteration.
To fix this issue, we make sure that the algorithm runs for at least two iterations.}

Both algorithms are written in Java 8 and share as much of their implementation as possible.
There is often a need to integrate over the bids $b_{\smi}$, e.g. when computing expected utilities or utility planes.
This integration is approximated using Monte Carlo sampling with common random numbers, as described in \cite{bosshard2017fastBNE}.
The number of samples used is \llgMCsamples{} for LLG and \llllggMCsamples{} for LLLLGG.

\subsection{Local-Local-Global (LLG)}

First, we consider the well-known LLG domain, a stylized setting with three bidders, two items and one bundle of interest per bidder.
\citet{bosshard2017fastBNE} measure the average runtime of their BNE algorithm in 16 variations of this setting for which analytical BNEs are known and can be compared against \cite{AusubelBaranov2013CoreSelectingAuctionsWithIncompleteInformation}.
All of the payment rules used are non-decreasing, including VCG-nearest, which is non-decreasing in LLG, even if not in general. 8 of the 16 variations have independent bidder distributions.
To achieve the required $\epsilon$ of \llgtargeteps{}, we employ \llggridsizeold{} verification points for the baseline algorithm, and we compute \llggridsize{} utility lines in our new algorithm.

\vspace{-0.03in}
\paragraph{Handling the Global Bidder.}
One peculiarity of LLG is that it is a dominant strategy for the global bidder to bid truthfully. This is easy to see, because from his perspective, a minimum revenue core-selecting auction is equivalent to a single item second price auction.
This de facto reduces the game to two bidders.
Unfortunately, our BNE algorithm requires all bidders to play piecewise constant strategies, so we cannot simply fix the global bidder to play the truthful strategy.

There are two possible ways to deal with this:
the first option is to force the global bidder to play a piecewise constant strategy, and to include his utility loss in the computation of $\varepsilon$.
The second option is to let the global bidder play truthfully, but bound the utility loss of the local bidders by pretending that the global bidder plays a piecewise constant strategy slightly above or below the truthful strategy.

We do the latter, to keep the runtime comparison between both algorithms as fair as possible.
Specifically, when computing the actual utilities for the local bidders, we let the global bidder bid above truth, but when computing the $i$-optimal utilities, we let him bid below truth.
In LLG, higher bids by the global bidder always lead to a decrease in the utilities of the local bidders.
Our approach thus subtracts an underestimate of the utility obtained from an overestimate of the highest utility that could possibly be obtained.
This produces a correct bound on the utility loss.

\vspace{-0.03in}
\paragraph{Results.}
The average runtimes are shown in the first row of Table~\ref{tab:runtimes}. As one can see, our algorithm has much better performance than the baseline, leading to a speedup of \llgspeedup{}x. Moreover, the equilibria we computed are close to the corresponding analytical BNEs: their $L_\infty$ distance is less than \llgBNEdistance{} in every case.

\newcounter{tempfootnote}
\setcounter{tempfootnote}{\value{footnote}}
\setcounter{footnote}{1}
\renewcommand{\thefootnote}{\alph{footnote}}
\renewcommand*{\thefootnote}{\fnsymbol{footnote}}

\begin{table}
    \setlength\tabcolsep{1.5pt}
    \begin{tabular}{|l||r|r|r|}
    \hline
    Domain \& Rule              & Runtime    & Runtime & Speedup \\
                         & (Baseline) & (Utility Planes)  & Factor \\

    \hline
    \hline
    LLG (8 variations)   &  0.0109    & 0.0017    & 6.45 \\
    \hline
    LLLLGG First Price   &  151.30    & 1.75    & 86.45 \\
    \hline
    LLLLGG Proxy         &  224.10    & 2.91    & 77.01 \\
    \hline
    LLLLGG Proportional\footnotemark  &  155.47    & 1.95    & 79.72 \\
    \hline
    \end{tabular}
    \caption{Runtimes for finding $\epsilon$-BNEs of auctions with non-decreasing payment rules, measured in core-hours. The $\epsilon$ reached is \llgtargeteps{} for LLG and \llllggtargeteps{} for LLLLGG. The baseline uses the algorithm from [Bosshard \emph{et al.}, 2017].}
    \label{tab:runtimes}
\vspace{-0.05in}
\end{table}

\footnotetext[2]{This rule did not converge in the baseline, only reaching \mbox{$\epsilon = 0.026$}. We still include it in the runtime comparison.}

\subsection{LLLLGG}

Next, we test our algorithm in a larger setting, called \mbox{LLLLGG}.
In this setting, we have 6 bidders and 8 items, with each bidder having two bundles of interest.
The bidders are split into two classes: there are 4 local and 2 global bidders, with strategies being symmetric within each class.
The value distributions of all bidders are independent of each other.
For a detailed definition, see \cite{bosshard2017fastBNE}.

To compute $\epsilon$-BNEs for this setting, we use a grid resolution of \llllgggridinterval{}.
Since local bidders' values are drawn from $[0,1]$ and global bidders' values are drawn from $[0,2]$, we compute utility lines corresponding to $\llllgggridsizeL{}^2$ and $\llllgggridsizeG{}^2$ equally spaced bids, respectively.
The baseline algorithm uses a grid of $120^2$ verification points for both the local and global bidders.

\setcounter{footnote}{\value{tempfootnote}}
\renewcommand{\thefootnote}{\arabic{footnote}}

\vspace{-0.03in}
\paragraph{Results.}
We run experiments for three non-decreasing payment rules: First-price, proxy and proportional. The runtimes are shown in Table~\ref{tab:runtimes}.\footnote{
    The runtime of these three rules is orders of magnitudes faster than what \citet{bosshard2017fastBNE} reported for VCG-nearest in their paper. This is expected because VCG-nearest is very expensive to evaluate, needing to compute the VCG point and then solve a quadratic program \cite{DayRaghavan2007FairPayments}.  
}
For all payment rules, the utility planes BNE algorithm finds a BNE over \llllggspeedup{} times faster than the baseline.

\section{Conclusion}

In this paper, we have introduced non-decreasing payment rules, and we have shown that this property has important consequences for incentives and algorithm design.
Importantly, the commonly used VCG-nearest rule is not non-decreasing and enables various kinds of manipulations. Our preliminary analysis suggests that rules that are not non-decreasing can be manipulated via overbidding.
In contrast, we conjecture that overbidding is never favorable in CAs with non-decreasing payment rules.
We have also developed the theory necessary to create the \emph{utility planes BNE algorithm}, which exploits the structural properties of non-decreasing payment rules to search for $\epsilon$-BNEs in a very efficient way. We have empirically found our algorithm to be highly performant, beating a recent state-of-the-art algorithm by multiple orders of magnitude. Thus, our BNE algorithm pushes the boundary on what problem sizes can be analyzed computationally when studying non-truthful payment rules for CAs. Overall, our results suggest that further analytical and algorithmic analysis of non-decreasing payment rules is a promising avenue for future research.

\section*{Acknowlegdements}

This paper is partially based on work done by the second author for his master's thesis at ETH Z\"urich.
We would like to thank our colleagues at ETHZ for making such collaborations possible.

\newpage
\bibliographystyle{named}

\ifthenelse{\boolean{withAppendix}}{
\newpage
\appendix

\section{Definitions of Payment Rules}

Let $W(b, x) = \sum_{i=1}^n b_i(x_i)$ be the reported social welfare achieved by an allocation $x$.
By extension, let $W_\smi(b_\smi, x) = \sum_{j \neq i} b_j(x_j)$ and let $W_L(b_L, x) = \sum_{j \in L} b_j(x_j)$.
Similarly, let $X_L(b_L)$ be the set of allocations $x$ maximizing $W_L(b_L, x)$.

\begin{definition}
    The \textbf{first price payment} is given by 
    $$p_i(b, x) := b_i(x_i).$$
\end{definition}

\begin{definition}
    The \textbf{VCG payment} is given by
    $$p_i(b, x) := W(b_\smi, X_{\smi}(b_{\smi})) - W_\smi(b_\smi, x).$$
\end{definition}

\begin{definition}
    The \textbf{core} is the set of all points $p(b, x)$ fulfilling
    \begin{align*}
        &\forall L \subseteq N: \\
        &\sum_{i \in N \setminus L} p_i(b, x) \geq W_L(b_L, X_L(b_L)) - W_L(b_L, x).
    \end{align*}
\end{definition}

\begin{definition}
    The \textbf{minimum-revenue core} is the set of all points $p(b, x)$ minimizing $\sum_{i \in N} p_i(b, x)$ subject to being in the core.
\end{definition}

\begin{definition}
    The \textbf{VCG-nearest payment} is the unique point $p(b, x)$ minimizing $||p_i(b, x) - q_i(b, x)||_2$ subject to being in the minimum-revenue core, where $q_i(b, x)$ is the VCG payment point.
\end{definition}

\begin{definition}
    The \textbf{proportional payment} is the unique point minimizing $\sum_{i \in N} p_i(b, x)$ subject to being in the core and being of the form
    $$p_i(b, x) := \alpha \cdot b_i(x_i)$$
    for some $\alpha \geq 0$.
\end{definition}

\begin{definition}
    \label{def:proxy}
    The \textbf{proxy payment} is the unique point minimizing $\sum_{i \in N} p_i(b, x)$ subject to being in the core and being of the form
    $$p_i(b, x) := \min \left[ \alpha, b_i(x_i) \right]$$
    for some $\alpha \geq 0$.
\end{definition}

Note that proxy and proportional payments are defined in \cite{AusubelBaranov2013CoreSelectingAuctionsWithIncompleteInformation} for LLG, and we generalize them in a simple way to other auction domains.
Our proxy payments should not be confused with the clock proxy auction, which is a mechanism proposed in the literature \cite{ausubel2002ascending} that reduces the multi-round combinatorial clock auction to a direct revelation mechanism. It achieves this by having each bidder communicate their valuation to a proxy agent, which then engages in ``straightforward bidding'' on their behalf.
We don't analyze the clock proxy auction in the present work, but have opted for the much simpler and computationally tractable payment rule given in Definition~\ref{def:proxy}. While algorithms to compute payments arising from the clock proxy auction have been proposed \cite{wurman2004computing,hoffman2006observations}, their outcome depends on many ambiguous details, such as exactly what straightforward bidding entails, the size of bidding increments, how ties are broken, and how interim bundle prices are computed.

\section{Proof of Theorem 1}

Recall that $X_{\smi}(K, b_{\smi})$ is a constrained allocation, where all goods not part of bundle $K$ are allocated to maximize reported social welfare across all bidders except $i$.

To simplify notation, we abbreviate the expression $b_i(K) + W(b_\smi, X_\smi(K, b_{\smi}))$ as $\tilde{W}(b_i, K)$, with the parameter $b_\smi$ being clear from context.

\begin{definition}
    A \textbf{cell} of bidder $i$'s action space is a connected region $S \subseteq \mathbb{R}_{\geq 0}^{r}$ where for any $b_{\smi}$ with positive probability of occurring in $s_{\smi}(v_{\smi})$, there exists an allocation $x$ that is efficient for each point in $S$, i.e.
    \begin{equation}
        \label{eq:cells}
        \forall b_{\smi} \sim s_{\smi}(v_{\smi}) \, \, \exists x \, \, \forall b_{i} \in S : x \in X(b_i, b_{\smi})
    \end{equation}

\end{definition}

Note that cells don't partition the action space in the strictest sense: a boundary between two cells is part of both cells.
The interiors of the cell are disjoint, however.

\begin{lemma}
    \label{lemma:1}
    A cell is a convex polytope.
\end{lemma}
\begin{proof}
    It is clear that $x \in X(b_i, b_\smi)$ is equivalent to
    \begin{align}
        \forall K \neq x_i : \tilde{W}(b_i, x_i) \geq \tilde{W}(b_i, K),
        \label{eq:welfareconstraints}
    \end{align}
    where the quantifier over $K$ includes all $r$ bundles of interest of bidder $i$, plus the empty bundle.
    Furthermore, note that
    \begin{align*}
         \tilde{W}(b_i, K) &\geq \tilde{W}(b_i, K') \\\quad &\Leftrightarrow\\
         b_i(K) - b_i(K') &\geq W(b_\smi, X_\smi(K', b_\smi)) - W(b_\smi, X_\smi(K, b_\smi))\\
         &\Leftrightarrow\\
         b_i(K) - b_i(K') &\geq c_{K, K'}
    \end{align*}
    where $c_{K,K'}$ is a constant independent of $b_i$.
    If $K'$ is the empty bundle, then the inequality simplifies to
    $$ b_i(K) \geq c_K.$$
    A cell is fully defined by these two types of inequalities.
\end{proof}

A pareto point of a cell is a point $b_i$ such that it's not possible to strictly decrease one coordinate of $b_i$ and weakly decrease all others, and still remain in the cell.\footnote{Note that this is a pareto point in the geometric sense, not the game-theoretic sense, though the terminology is of course related.}
\begin{lemma}
    \label{lemma:3}
    A cell has a unique pareto point, which is a vertex.
\end{lemma}
\begin{proof}
    Assume that we have two pareto points $b_i$ and $b_i'$ of the same cell $S$.
    Consider the point $b_i''$, defined as the coordinatewise minimum between these two points, i.e.
    $$\forall K : b_i''(K) := \min(b_i(K), b_i'(K)),$$
    which implies that
    \begin{align}
        \label{eq:p4} &\forall K : \tilde{W}(b_i'', K) = \min( \tilde{W}(b_i, K), \tilde{W}(b_i', K) ).
    \end{align}
    Since $b_i''$ pareto dominates both $b_i$ and $b_i'$, it must be contained in a cell other than $S$ (otherwise it would contradict our assumption that $b_i$ and $b_i'$ are pareto points of $S$).
    It follows from (\ref{eq:cells}) and (\ref{eq:welfareconstraints}) that there exists some bid of other bidders $b_\smi$ and an allocation $x$ such that
    \begin{align}
        \label{eq:p1} &\forall K \neq x_i : \tilde{W}(b_i, x_i) \geq \tilde{W}(b_i, K),\\
        \label{eq:p2} &\forall K \neq x_i : \tilde{W}(b_i', x_i) \geq \tilde{W}(b_i', K),\\
        \label{eq:p3} &\exists K \neq x_i : \tilde{W}(b_i'', x_i) < \tilde{W}(b_i'', K).
    \end{align}
    Chaining together (\ref{eq:p3}), (\ref{eq:p4}) and (\ref{eq:p1}) we get that 
    \begin{align*}
        &\tilde{W}(b_i'', x_i) < \tilde{W}(b_i'', K) \leq \tilde{W}(b_i, K) \leq \tilde{W}(b_i, x_i).
    \end{align*}
    Analogously, from (\ref{eq:p3}), (\ref{eq:p4}) and (\ref{eq:p2}) we get that
    \begin{align*}
        &\tilde{W}(b_i'', x_i) < \tilde{W}(b_i'', K) \leq \tilde{W}(b_i', K) \leq \tilde{W}(b_i', x_i),
    \end{align*}
    which is a contradiction to (\ref{eq:p4}).
    
    It remains to show that the unique pareto point $b_i$ is a vertex.
    If $b_i$ lies strictly in the interior of some face $F$, there exists at least one direction $d$ such that the points $b_i + \epsilon \cdot d$ and $b_i - \epsilon \cdot d$ are also in the interior of $F$ (for small enough $\epsilon$).
    If d has only positive or negative coordinates, then the point $b_i - \epsilon \cdot d$ respectively $b_i + \epsilon \cdot d$ pareto dominates $b_i$.
    If d has a mix of positive and negative coordinates, then $b_i + \epsilon \cdot d$ neither dominates nor is dominated by $b_i$, so it is either a pareto point, or dominated by some other pareto point.
\end{proof}

\begin{lemma}
    \label{lemma:utilbound}
    In a CA with a non-decreasing payment rule and where bidders have piecewise constant strategies, we have that for any $v_i$, there exists a cell vertex $b_i^*$ such that for any bid $b_{i}$ in $i$'s action space
    $$\euo_i(v_i, b_i^*) \geq \euo_i(v_i, b_i).$$
\end{lemma}
\begin{proof}
    Let $\pi$ be the permutation maximizing $\euo_i(v_i, b_i)$.
    For $\delta \rightarrow 0$, $b_i + \delta \pi$ is in the interior of some cell $S$, because all the coordinates of $\pi$ are distinct, and thus $\pi$ is not parallel to any of the constraints given by (\ref{eq:welfareconstraints}).
    Let $b_i^*$ be the pareto point of $S$, and $\pi^*$ a permutation such that $b_i^* + \delta \pi^*$ is also in the interior of $S$.
    Player $i$'s probability of winning each bundle is identical for $b_i^* + \delta \pi^*$ and $b_i + \delta \pi$, but the expected payment is weakly smaller for $b_i^*$ than for $b_i$, because the payment rule is non-decreasing.
\end{proof}

\begin{lemma}
    \label{lemma:utilbound2}
    In a CA with a non-decreasing payment rule and where bidders have piecewise constant strategies, for any value $v_{i}$ and bid $b_{i}$ in $i$'s action space, we have that
    $$\euo_i(v_i, b_i) \geq \eu_i(v_i, b_i).$$
\end{lemma}
\begin{proof}
    We show that
    $$\euo_i(v_i, b_i) \geq \lim_{\delta \rightarrow 0} \eup_i(v_i, b_i, \mathbf{1}) \geq \eu_i(v_i, b_i).$$
    
    The second inequality follows from free disposal, because bidder $i$ has weakly higher probability of winning each non-empty bundle under the allocation $X(b_i +\delta \mathbf{1}, b_{\smi})$, and the ratio of winning probabilities between non-empty bundles stays the same.
    
    For the first inequality, we show that $\lim_{\delta \rightarrow 0} \eup_i(v_i, b_i, x) \leq \euo_i(v_i, b_i)$ for any vector $x$ in the positive orthant.
    We convert $x$ into a strict preference ordering by changing one coordinate at a time, while weakly increasing the expected utility.
    This proves the claim because $\euo_i$ maximizes over all strict preference orderings.
    Consider two coordinates $x_a = x_b$.
    The allocation $X(b_i + \delta x, b_{\smi})$ has a certain probability of tying between bundles $a$ and $b$. 
    By slightly increasing either $x_a$ or $x_b$, we can weakly increase the expected value, because the expected value under the original $x$ is an average of these two results, and thus can't be strictly better than both.
\end{proof}

\begin{proof}[Proof of Theorem 1]
    Combine Lemmas \ref{lemma:utilbound} and \ref{lemma:utilbound2}.
\end{proof}

}{}

\clearpage
\newpage
\includepdf[pages=-]{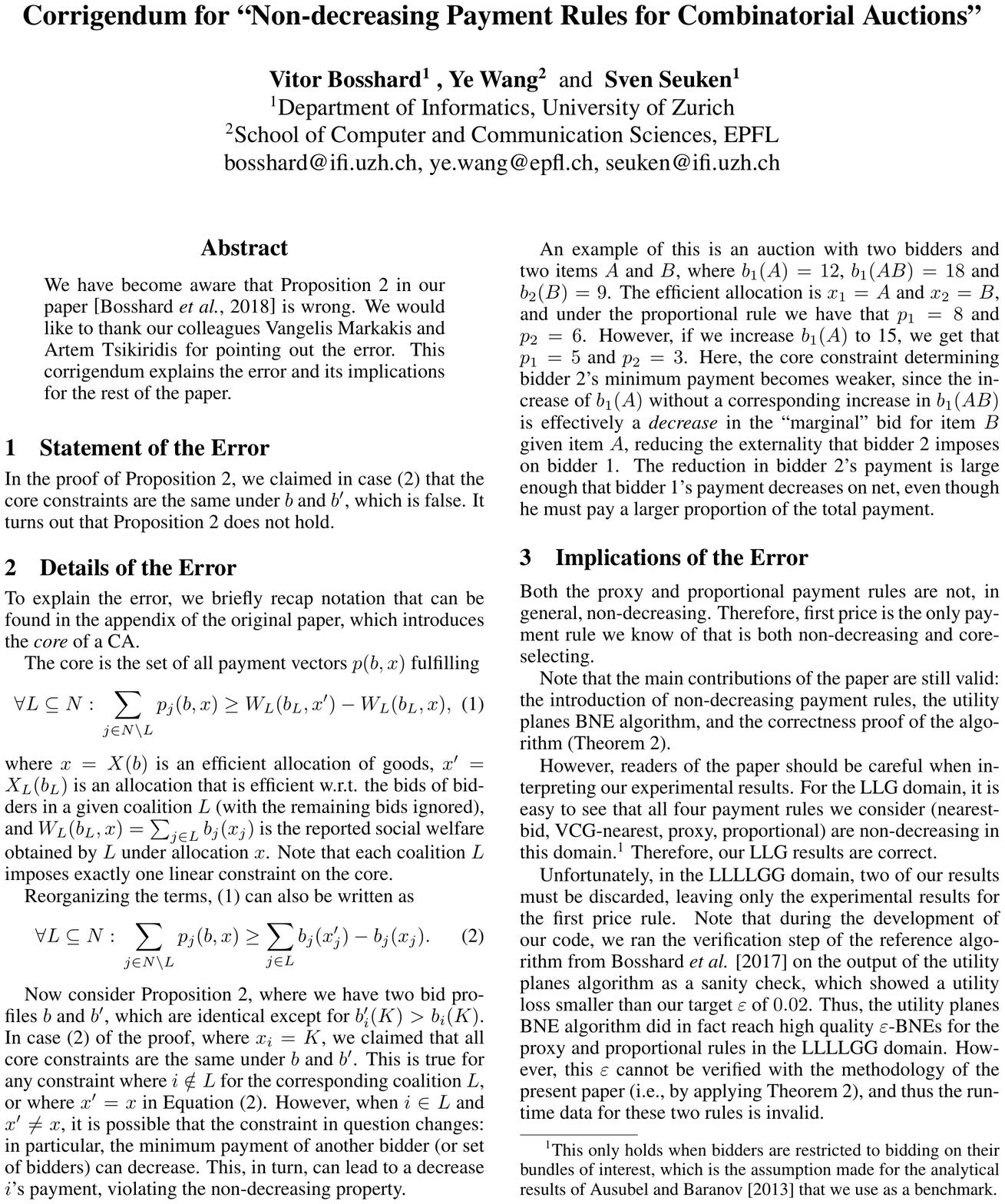}

\end{document}


\maketitle

\begin{abstract}
We have become aware that Proposition~2 in our paper \cite{bosshard2018nondecreasing} is wrong.
We would like to thank our colleagues Vangelis Markakis and Artem Tsikiridis for pointing out the error.
This corrigendum explains the error and its implications for the rest of the paper.
\end{abstract}

\section{Statement of the Error}

In the proof of Proposition~2, we claimed in case (2) that the core constraints are the same under $b$ and $b'$, which is false.
It turns out that Proposition~2 does not hold.

\section{Details of the Error}

To explain the error, we briefly recap notation that can be found in the appendix of the original paper, which introduces the \emph{core} of a CA.

The core is the set of all payment vectors $p(b, x)$ fulfilling
\begin{equation}
    \label{eq}
    \forall L \subseteq N: \\
    \sum_{j \in N \setminus L} p_j(b, x) \geq W_L(b_L, x') - W_L(b_L, x),
\end{equation}
where $x = X(b)$ is an efficient allocation of goods, $x' = X_L(b_L)$ is an allocation that is efficient w.r.t. the bids of bidders in a given coalition $L$ (with the remaining bids ignored), and $W_L(b_L, x) = \sum_{j \in L} b_j(x_j)$ is the reported social welfare obtained by $L$ under allocation $x$.
Note that each coalition $L$ imposes exactly one linear constraint on the core.

Reorganizing the terms, \eqref{eq} can also be written as
\begin{equation}
    \label{eq2}
    \forall L \subseteq N: \\
    \sum_{j \in N \setminus L} p_j(b, x) \geq \sum_{j \in L} b_j(x'_j) - b_j(x_j).
\end{equation}

Now consider Proposition~2, where we have two bid profiles $b$ and $b'$, which are identical except for $b'_i(K) > b_i(K)$.
In case (2) of the proof, where $x_i = K$, we claimed that all core constraints are the same under $b$ and $b'$.
This is true for any constraint where $i \notin L$ for the corresponding coalition $L$, or where $x' = x$ in Equation~\eqref{eq2}.
However, when $i \in L$ and $x' \neq x$, it is possible that the constraint in question changes: in particular, the minimum payment of another bidder (or set of bidders) can decrease.
This, in turn, can lead to a decrease $i$'s payment, violating the non-decreasing property.

An example of this is an auction with two bidders and two items $A$ and $B$, where $b_1(A) = 12$, $b_1(AB) = 18$ and $b_2(B) = 9$.
The efficient allocation is $x_1 = A$ and $x_2 = B$, and under the proportional rule we have that $p_1 = 8$ and $p_2 = 6$.
However, if we increase $b_1(A)$ to 15, we get that $p_1 = 5$ and $p_2 = 3$.
Here, the core constraint determining bidder 2's minimum payment becomes weaker, since the increase of $b_1(A)$ without a corresponding increase in $b_1(AB)$ is effectively a \emph{decrease} in the ``marginal'' bid for item $B$ given item $A$, reducing the externality that bidder 2 imposes on bidder 1. The reduction in bidder 2's payment is large enough that bidder 1's payment decreases on net, even though he must pay a larger proportion of the total payment.

\section{Implications of the Error}

Both the proxy and proportional payment rules are not, in general, non-decreasing.
Therefore, first price is the only payment rule we know of that is both non-decreasing and core-selecting.

Note that the main contributions of the paper are still valid: the introduction of non-decreasing payment rules, the utility planes BNE algorithm, and the correctness proof of the algorithm (Theorem~2).

However, readers of the paper should be careful when interpreting our experimental results.
For the LLG domain, it is easy to see that all four payment rules we consider (nearest-bid, VCG-nearest, proxy, proportional) are non-decreasing in this domain.\footnote{This only holds when bidders are restricted to bidding on their bundles of interest, which is the assumption made for the analytical results of \citet{AusubelBaranov2013CoreSelectingAuctionsWithIncompleteInformation} that we use as a benchmark.}
Therefore, our LLG results are correct.

Unfortunately, in the LLLLGG domain, two of our results must be discarded, leaving only the experimental results for the first price rule.
Note that during the development of our code, we ran the verification step of the reference algorithm from \citet{bosshard2017fastBNE} on the output of the utility planes algorithm as a sanity check, which showed a utility loss smaller than our target $\varepsilon$ of $0.02$.
Thus, the utility planes BNE algorithm did in fact reach high quality $\varepsilon$-BNEs for the proxy and proportional rules in the LLLLGG domain.
However, this $\varepsilon$ cannot be verified with the methodology of the present paper (i.e., by applying Theorem~2), and thus the runtime data for these two rules is invalid.

\bibliographystyle{named}